\documentclass{article}
\usepackage{spconf,amsmath,graphicx}
\usepackage{array,delarray,xspace,wrapfig}
\usepackage{algorithmic}
\usepackage{algorithm}
\usepackage{cite}
\usepackage{amsthm}
\newcommand{\sign}{\mathrm{sign}}
\usepackage{amsfonts,amssymb,mathrsfs,enumitem, subcaption,relsize,scalerel}
\usepackage{hyperref}

\newtheorem{proposition}{Proposition}

\usepackage[dvipsnames]{xcolor}
\title{Multiplication-Avoiding Variant of Power Iteration with Applications}
\name{Hongyi Pan \quad Diaa Badawi \quad Runxuan Miao\quad Erdem Koyuncu\quad Ahmet Enis Cetin\thanks{This work was supported in part by Army Research Lab
(ARL) under Grant W911NF-21-2-0272, National Science Foundation (NSF) under Grants 1934915 and CCF-1814717, and by an award from the University of Illinois at Chicago Discovery Partners Institute Seed Funding Program.}}
\address{Department of Electrical and Computer Engineering, University of Illinois at Chicago, Chicago, IL}
\begin{document}
%\ninept
%
\maketitle
\small
\begin{abstract}
  Power iteration is a fundamental algorithm in data analysis. It extracts the eigenvector corresponding to the largest eigenvalue of a given matrix. Applications include ranking algorithms, principal component analysis (PCA), among many others. Certain use cases may benefit from alternate, non-linear power methods with low complexity.\ In this paper, we introduce multiplication-avoiding power iteration (MAPI). MAPI replaces the standard $\ell_2$ inner products that appear at the regular power iteration (RPI) with multiplication-free vector products, which are Mercer-type kernels that induce the $\ell_1$ norm. For an $n\times n$ matrix, MAPI requires $n$ multiplications, while RPI needs $n^2$ multiplications per iteration. Therefore, MAPI provides a significant reduction of the number of multiplication operations, which are known to be costly in terms of energy consumption. We provide applications of MAPI to PCA-based image reconstruction as well as to graph-based ranking algorithms. When compared to RPI, MAPI not only typically converges much faster, but also provides superior performance.
% In this article, we propose a multiplication avoiding power iteration method based on a kernel-related with the $\ell_1$-norm. The kernel is computed using sign calculations and absolute value comparisons. It is energy efficient as the kernel computation does not require any multiplications. We use Kernel Power Iteration (KPI) method in a number of applications. We compare the KPI-based PCA with the regular PCA and $\ell_1$-PCA in an image reconstitution method. The KPI-PCA reaches 0.90 dB higher than the regular PCA with conventional power iteration on reconstructed images. Then, to show its convergence speed faster than the conventional power iteration, we test the algorithm on a synthetic dataset. Finally, we show that our min power iteration can be implemented into the Google PageRank algorithm to accelerate the convergence.
\end{abstract}

\section{Introduction}\vspace{-5pt}
\label{secIntro}
Let $\mathbf{A}$ be a diagonalizable matrix and $\mathbf{b}$ be some initial vector. Power iteration is described by the update equation or recurrence relation
\begin{align}
\label{poweriteration}
\mathbf{b} \leftarrow \frac{\mathbf{A} \mathbf{b}}{\|\mathbf{A} \mathbf{b}\|}, 
\end{align}
where $\|\cdot\|$ represents the Euclidean norm. It is well-known that (\ref{poweriteration}) converges to the eigenvector of $\mathbf{A}$ corresponding to the dominant eigenvalue. Note that the normalization can be omitted, i.e., $\mathbf{b} \leftarrow \mathbf{A} \mathbf{b}$ results in the same direction as the dominant eigenvector.

A primary application of the power iteration (\ref{poweriteration}) is conventional (or, $L^2$) principal component analysis (PCA): Suppose that we collect members of a zero-mean $D$-dimensional dataset $\{\mathbf{x}_1,\ldots,\mathbf{x}_N\}\subset\mathbb{R}^D$ to a $D\times N$ matrix $\mathbf{X}=[\mathbf{x}_1 \ \mathbf{x}_2 \ ... \  \mathbf{x}_N]\in \mathbb{R}^{D\times N}$. Consider the corresponding sample covariance matrix
\begin{align}
\label{l2covmatrix}
\mathbf{C} = \frac{1}{N} \mathbf{X} \mathbf{X}^T.
\end{align}
The first principal vector, say $\mathbf{p}_1$, is then the dominant eigenvector of $\mathbf{C}$. The eigenvector $\mathbf{p}_1$ can be extracted via the power iteration in (\ref{poweriteration}) applied with the substitution $\mathbf{A} = \mathbf{C}$. The $i$th principal vector can then be extracted via power iteration on $(\mathbf{I} - \sum_{k=1}^{i-1}\mathbf{p}_k\mathbf{p}_k^T)\mathbf{C}$. 
% \begin{align}
% \label{pcapoweriter}
% \mathbf{b} \leftarrow \frac{\mathbf{C}\mathbf{b}}{\|\mathbf{C}\mathbf{b}\|}. 
% \end{align}

% There are two fundamental motivations behind our proposal of a multiplication-avoiding variant of power iteration. One is to achieve increased robustness. In fact, c

In the process of computing the PCA of $\mathbf{C}$ through (\ref{poweriteration}), one effectively computes the Euclidean inner product of the candidate principal vector $\mathbf{b}$ with the dataset elements $\mathbf{x}_1,\ldots,\mathbf{x}_N$ at each iteration. The fundamental underlying operation for each such Euclidean inner product is the multiplication of the components of $\mathbf{b}$ with the components of $\mathbf{x}_i$. On the other hand, it is well-known that multiplication operation will overamplify the effects of outliers or noise in the dataset. To increase robustness, our idea is to replace the Euclidean inner products that appear in (\ref{poweriteration}) with Multiplication Avoiding Vector Products (MAVPs), which were originally developed in the context of neural networks  \cite{tuna2009image, afrasiyabi2017energy,afrasiyabi2018non, pan2019additive, ergen2019energy, nasrin2021mf}. Two of the MAVPs satisfy the Mercer-type kernel requirement \cite{pan2021robust,scholkopf:kpca}. 
The resulting multiplication-avoiding power iteration (MAPI) becomes a power iteration in the Reproducing Kernel Hilbert Space (RKHS) defined by the kernel. Our MAVPs utilize the sign information of vector components to preserve the correlative property of multiplication; see \cite{eweda1994comparison} for another work that utilizes the related idea of sign regression.

Energy efficiency and improved computational complexity provides another major motivation for utilizing MAPIs in lieu of ordinary power iteration. In fact, MAVPs rely only on minimum operations and sign changes, and avoid the energy consuming multiplication operation. Therefore, compared to an Euclidean inner product, an MAVP can be executed 
%much faster and 
in an energy efficient manner in many processors. The same benefits transfer to MAPI, which utilize MAVPs.

We reemphasize that RPI and MAPI will result in different vectors upon convergence; i.e. they are not equivalent. We expect MAPI to be more robust than RPI as it utilizes MAVPs that induce the $\ell_1$-norm. We also expect the MAPI to be more energy- and computationally-efficient as it avoids the multiplication operation.

Kernel-PCA was introduced by Scholkopf et al. \cite{scholkopf:kpca,scholkopf1998nonlinear}. It has been applied to many signal and image processing problems \cite{goldberg2007kernel,nasrabadi2013hyperspectral,van2012kernel,hoffmann2007kernel}.
Related work includes recursive $\ell_1$-PCA and PCA methods using the similarity measures related with the $\ell_1$-norm \cite{markopoulos2014optimal, markopoulos2017efficient, markopoulos2017indoor, markopoulos2018adaptive,pan2021robust}. However all of the above mentioned methods are computationally costly in large covariance matrices because they either require the eigen-decomposition of the covariance matrix or the solution of a complex optimization problem. 
To the best of our knowledge, this paper is the first paper describing power iterations in the kernel domain. This is probably due to the fact that other kernel PCA methods require costly kernel computations compared to the regular vector dot product \cite{hoffmann2007kernel, xiao2014model, kim2017l1, varon2015noise, battaglino2020generalization}. On the other hand, MAPI kernel operations are more energy efficient than the dot product.
In the context of PCA, MAPI can not only improve robustness but provide significant improvements in computational complexity. MAPI provides similar benefits in other applications in which the power iteration is used. Examples will be provided throughout the paper.

The rest of the paper is organized as follows: In Section \ref{secMethodology}, we formally introduce the MAPI. In Section \ref{secNumerical}, we describe applications of MAPI and report corresponding results. Finally, in Section \ref{secConclusions}, we draw our main conclusions.

\section{Methodology}
\label{secMethodology}
In this section, we provide an overview of the MAVPs, and formally introduce the MAPI method that utilizes MAVPs. 
\subsection{Multiplication-Avoiding Vector Products (MAVPs)}
MAVPs were first studied in \cite{tuna2009image, akbas2014l1} to develop a robust region covariance matrix. They were used in computationally efficient neural networks. In this work, we will utilize the MAVP as
\begin{align}
	\label{odotproduct}
	\mathbf{w}^T \odot \mathbf{x}  & \triangleq \sum_{i=1}^n \sign(w_i x_i) \min(|w_i|, |x_i|).
\end{align}
We define another related dot-product as follows:
%and its variation:
\begin{align}
	\label{odotproduct2}
	\mathbf{w}^T\odot_m \mathbf{x} \triangleq \sum_{i=1}^n \mathbf{1}\left(\sign(w_i) = \sign(x_i)\right) \min(|w_i|, |x_i|)
\end{align}
where $\mathbf{1}(\cdot)$ is the indicator function. We call (\ref{odotproduct}) as min1 operation and (\ref{odotproduct2}) as min2 operation. In the following sessions, we will denote the two min-operations as $\oplus = \{\odot, \odot_m\}$. 
Note that unlike an ordinary Euclidean inner product, (\ref{odotproduct}) and (\ref{odotproduct2}) do not contain any multiplication operations. Energy efficiency of $\oplus$ operation varies from processor to processor.
For example, multiplication-based regular dot-product operation consumes about 4 times more energy compared to the multiplication avoiding dot product operations defined in (\ref{odotproduct}) and (\ref{odotproduct2}) in compute-in-memory (CIM) implementation at 1 GHz operating frequency \cite{nasrin2021mf}. The MAVPs induce the $\ell_1$-norm as $\mathbf{x}^T \oplus \mathbf{x} = \sum_{i=1}^n \min(|x_i|, |x_i|) = \|\mathbf{x}\|_1$. 

The MAVPs can be extended to matrix multiplications as follows: Let $\mathbf{W} = [\mathbf{w}_1\cdots\mathbf{w}_m] \in \mathbb{R}^{n\times m}$ and $\mathbf{X} = [\mathbf{x}_1\cdots\mathbf{x}_p] \in \mathbb{R}^{n\times p}$ be arbitrary matrices. We define $\mathbf{W}^T \oplus \mathbf{X}$ as an $n\times p$ matrix whose entry in the $i$th row, $j$th column is $\mathbf{w}_i^T\oplus \mathbf{x}_p$.
% \begin{align}
% 	%\label{matmulgen}
% 	\mathbf{W}^T \!\oplus\! \mathbf{X} \triangleq 
% 	\begin{bmatrix}
% 		\mathbf{w}_1^T\oplus \mathbf{x}_1&\mathbf{w}_1^T\oplus \mathbf{x}_2&\dots&\mathbf{w}_1^T\oplus \mathbf{x}_p\!\!\!\!\!\!\\
% 		\mathbf{w}_2^T\oplus \mathbf{x}_1&\mathbf{w}_2^T\oplus \mathbf{x}_2&\dots&\mathbf{w}_2^T\oplus \mathbf{x}_p\!\!\!\!\!\!\\
% 		\vdots&\vdots&\ddots&\vdots&\\
% 		\mathbf{w}_m^T\oplus \mathbf{x}_1&\mathbf{w}_m^T\oplus \mathbf{x}_2&\dots&\mathbf{w}_m^T\oplus \mathbf{x}_p\!\!\!\!\!\!
% 	\end{bmatrix}
% 	\label{eq: matrix}
% \end{align}
Thus, the definition is similar to matrix multiplication $\mathbf{W}^T\mathbf{X}$ by only changing the element-wise product to element-wise MAVP. 

% MAVPs provide a variant of PCA

% PCA and variant methods are dimension reduction techniques  that rely on orthogonal transformations \cite{pearson:pca, hotelling:pca, pca:tutorial}. Specifically, suppose that we collect members of a zero-mean $D$-dimensional dataset $\{\mathbf{x}_1,\ldots,\mathbf{x}_N\}\subset\mathbb{R}^D$ to a $D\times N$ matrix $X=[\mathbf{x}_1 \ \mathbf{x}_2 \ ... \  \mathbf{x}_N]\in \mathbb{R}^{D\times N}$. In the conventional $L^2$-PCA, one first ``extracts'' the first principal vector via
% \begin{align}
% \label{l2pcaextract}
% \mathbf{w}_{L^2}^{(1)} \triangleq \arg\max_{\mathbf{w}:\|\mathbf{w}\| = 1} \textstyle\sum_{i=1}^N (\mathbf{w}^T \mathbf{x}_i)^2
% \end{align}
% where $(\cdot)^T$ is the transpose, and $\|\mathbf{w}\|\triangleq \sqrt{\mathbf{w}^T \mathbf{w}}$ is the $\ell_2$-norm. The solution of (\ref{l2pcaextract})  is the eigenvector corresponding to the largest eigenvalue of the sample covariance matrix
% \begin{align}
% \label{l2covmatrix}
% \mathbf{C} = \frac{1}{N} \mathbf{X X}^T
% \end{align}

Recall from Section \ref{secIntro} that given a dataset $\mathbf{X}=[\mathbf{x}_1 \ \mathbf{x}_2 \ ... \  \mathbf{x}_N]\in \mathbb{R}^{D\times N}$, the eigendecomposition of $\mathbf{C} = \frac{1}{N} \mathbf{X} \mathbf{X}^T$ yields the ordinary PCA. Alternatively, one can construct the sample covariance matrix through MAVPs. In particular, in \cite{pan2021robust}, we considered the eigendecomposition of the ``min-covariance matrix''
\begin{equation}
\label{l1samplecovariance}
\mathbf{A} = \tfrac{1}{N}\mathbf{X} \oplus \mathbf{X}^T
\end{equation}

We have shown in \cite{pan2021robust} that the resulting ``Min-PCA'' provides better resilience against impulsive noise than regular PCA in image reconstruction experiments.

\subsection{Multiplication-Avoiding Power Iteration (MAPI)}
We are now ready to introduce the MAPI. We replace the standard products in (\ref{poweriteration}) with a MAVP. To further reduce the computational complexity, we replace the normalization by $\ell_2$-norm by a normalization by $\ell_1$-norm. Note that the calculation of the $\ell_2$-norm requires multiplications while $\ell_1$-norm can be calculated via additions only. Our revisions yield the iteration
\begin{equation}
\mathbf{w} \leftarrow \frac{\mathbf{A} \oplus \mathbf{w}}{|| \mathbf{A} \oplus \mathbf{w} ||_1}
\label{kernelpoweriter}
\end{equation}
We have observed that such a change of normalization does not effect the final performance greatly; also see \cite{lin2010power, arikan1994adaptive} for studies that utilize an $\ell_1$ normalization. The final MAPI algorithm is shown in Algorithm~\ref{al: MAPI}. For applications to PCA, we can normalize the final $\mathbf{w}_t$ by its $l_2$-norm for extraction of subsequent principal vectors. 
% In general, if $\mathbf{A}$ is a diagonalizable matrix with a dominant eigenvalue, then there exists a
% nonzero vector $\mathbf{w}_0$ such that the sequence of vectors defined by
% $
% \mathbf{w}_{t+1} = \mathbf{A  w}_t , \ \ t=0, 1, 2,...  $
% approaches a multiple of the dominant eigenvector of $\mathbf{A}$ \cite{golub2013matrix,souza2017fixed,shi2011accelerated,kim2019simple}. In our case, the normalized
% power iterations
% \begin{equation}
% \mathbf{w}_{t+1} = \frac{\mathbf{A  w}_t}{|| \mathbf{A  w}_t ||}, \ \ t=0, 1, 2,...
% \label{kernelpoweriter1}
% \end{equation}
% converge to the dominant eigenvector of the MF-Kernel matrices defined in the previous section. We can also use the fast algorithm introduced in \cite{kim2019simple} to find the eigenvectors of the kernel matrices because the fast algorithm requires only the kernel matrix $A$ as input. 

% We also experimentally studied the following (almost) multiplication-free ``power iterations" as a preliminary work:

% where $\mathbf{A}$ is the min-covariance matrix defined in Eq. (\ref{l1samplecovariance}). We can normalize the final $\mathbf{w}_t$ by its $l_2$-norm after iteration ends.

\begin{algorithm}[htbp]
	\caption{MAPI}
	\label{al: MAPI}
	\begin{algorithmic}[1]
		\renewcommand{\algorithmicrequire}{\textbf{Input:}}
		\renewcommand{\algorithmicensure}{\textbf{Output:}}
		
		\REQUIRE $\mathbf{A} \in \mathbb{R}^{N\times N}$, iteration times $T$.
	%	\ENSURE  Dominant-pseudo-eigenvector $\mathbf{w}_T \in \mathbb{R}^{N\times1}$.
		
		\STATE Initialize $\mathbf{w}_0$ as a random vector with $||\mathbf{w}_0||=1$;
		\FOR{$t=0, 1, ..., T-1$} 
		\STATE $\mathbf{w}_{t+1} = \mathbf{A}\oplus \mathbf{w}_t$;
		\STATE $\mathbf{w}_{t+1} \leftarrow \mathbf{w}_{t+1}/||\mathbf{w}_{t+1}||_1$;
		\ENDFOR 
		\STATE (optional) $\mathbf{w}_T = \mathbf{w}_T/||\mathbf{w}_T||_2$;
		\RETURN  dominant pseudo-eigenvector $\mathbf{w}_T$.\\
	\end{algorithmic}
\end{algorithm}
Each step of the MAPI defined in (\ref{kernelpoweriter}) corresponds to
%${\bf A: R^N \rightarrow R^N}$
a transformation in the RKHS. As a result the convergence of MAPI depends on the matrix ${\bf A}$. Since we normalize the iterations in (8) at each step, the iterates are bounded and they satisfy $||{\bf w}_t ||_1 =1$.
We have observed that the MAPI in Algorithm~\ref{al: MAPI} converges in all experiments that we have tried and the resulting vector can be used in practical applications. 
For certain MF operators, convergence can even be proved analytically. In particular, let 
\begin{align}
\label{oplusoperator}
    \mathbf{w}^T \lozenge \mathbf{x} = \sum_{i=1}^n \left( \mathrm{sign}(x_i) y_i + x_i \mathrm{sign}(y_i) \right)
\end{align}
The operator $\lozenge$ has proved very useful for neural network applications \cite{afrasiyabi2017energy}. The convergence behavior of the MAPI is particularly simple, at least in the case of positive entries. Let $\mathbf{A} = [\mathbf{a}_1^T \cdots \mathbf{a}_n^T]$.
\begin{proposition}
For any initial vector $\mathbf{w}$ and matrix $\mathbf{A}$ with positive entries, the iterations $\mathbf{w} \!\leftarrow\! \frac{\mathbf{A} \lozenge \mathbf{w}}{|| \mathbf{A} \lozenge \mathbf{w} ||_1}$ converges to the vector
\begin{align}
    \frac{[1+\|\mathbf{a}_1\|_1 \cdots 1+\|\mathbf{a}_n\|_1]^T}{\sum_{i=1}^n (1+\|\mathbf{a}_i\|_1)}
\end{align}
in at most two iterations. 
\end{proposition}
\begin{proof}
The proof is by calculation. After the first iteration, one obtains a vector $\mathbf{c} = \mathbf{A} \lozenge \mathbf{w} = [c_1 \cdots c_N]^T$ whose $i$th component is 
\begin{align}
    c_i \triangleq \frac{\mathbf{a}_1^T \lozenge \mathbf{b}}{\sum_{j=1}^n(\mathbf{a}_j^T \lozenge \mathbf{b})} = \frac{\|\mathbf{a}_i\|_1 + \|\mathbf{b}\|_1}{\sum_{j=1}(\|\mathbf{a}_j\|_1 + \|\mathbf{b}\|_1)}.
\end{align}
Note that $\|\mathbf{c}\|_1 = 1$. Therefore, after the second iteration, we obtain a vector whose $i$th component is
\begin{align}
    \frac{\|\mathbf{a}_i\|_1 + \|\mathbf{c}\|_1}{\sum_{j=1}(\|\mathbf{a}_j\|_1 + \|\mathbf{c}\|_1)}
\end{align}
The result follows since $\|\mathbf{c}\|_1 = 1$.
\end{proof}

The general case for different operators and arbitrary entries for matrices is left as future work.

Note that MAPI only requires $N$ divisions per iteration as opposed to the RPI, which requires $N^2$ multiplications and $N$ divisions. Therefore, we expect MAPI to consume significantly less energy compared to RPI in most processors.

% The standard procedure for subsequent principal vectors $\mathbf{w}_j$ ($j\ge 2$) is to do power iteration on $\mathbf{C}-\sum_{k=1}^{j-1}\mathbf{w}_k\mathbf{w}_k^T\mathbf{C}$.
% Another possibility is to use $\mathbf{C}-\sum_{k=1}^{j-1}(\mathbf{w}_k \odot \mathbf{w}_k^T)\mathbf{C}$
% to estimate the subsequent principle vectors.
\vspace{-8pt}
\section{Applications and Numerical Results}
\label{secNumerical}
\vspace{-8pt}
\subsection{Image Reconstruction Example}

 We consider image reconstruction example studied in \cite{markopoulos2014optimal}. Our experiment follows the same structure in \cite{markopoulos2014optimal, pan2021robust}.
 %but instead of using eigendecomposition, 
 We use regular power iteration (RPI) and MAPI to compare the image reconstruction results of the MAPI. In this experiment, the image size $D\times D=128\times128$, and the pixel values are in the range of $[0, 1]$. Suppose that we want to reconstruct a gray-scale image (Fig.~\ref{fig: Image Reconstruction original_3b}) from its $N=10$ occluded versions. As Fig.~\ref{fig: Image Reconstruction occluded_3b} shows, the corrupted images are created by partitioning the original image into sixteen tiles of size $32\times32$ and replacing three arbitrarily selected tiles by $32\times32$ gray-scale noise patches.  The pixel values of the noise patches are independent and identically distributed random variables that are uniform on $[0,1]$.  The reconstruction algorithm is described in Algorithm~\ref{al: image reconstruction}. In particular, in Step 5 and Step 6, we obtain the first two dominant generalized eigenvectors of the min-covariance matrix of images via MAPI. We then reconstruct the image using these two generalized eigenvectors in Line 8.  
 %For RPI, we substitute all $\oplus$ operators with  standard multiplication and perform $\ell_2$-normalization.
 
 \begin{algorithm}[htbp]
	\caption{Image Reconstruction via MAPI}
	\label{al: image reconstruction}
	\begin{algorithmic}[1]
		\renewcommand{\algorithmicrequire}{\textbf{Input:}}
		\renewcommand{\algorithmicensure}{\textbf{Output:}}
		
		\REQUIRE $N$ corrupted versions of an image: $\mathbf{I}_1, \mathbf{I}_2, ..., \mathbf{I}_N \in \mathbb{R}^{D\times D}$, pixels are in range of $[0, 1]$.
		%\ENSURE  Reconstructed image $\mathbf{\hat{I}}$.
		\STATE Reshape $\mathbf{I}_i$ to vector $\mathbf{v}_i\in \mathbb{R}^{D^2\times 1}$ for $i=1,\ldots,N$;
		\STATE $\mathbf{V} = [\mathbf{v}_1\ \mathbf{v}_2\ ...\ \mathbf{v}_N] \in \mathbb{R}^{D^2\times N}$;
		\STATE Reduce mean: $\mathbf{m} = \text{mean}(\mathbf{V}, 2)\! \in\! \mathbb{R}^{D^2\times 1}, \mathbf{V} \!=\! \mathbf{V} - \mathbf{m}$;
		\STATE Construct the min-covariance matrix  $\mathbf{C}=\frac{1}{N-1}\mathbf{V}\oplus\mathbf{V}^T \in \mathbb{R}^{D^2\times D^2}$ of $\mathbf{V}$;
		\STATE Perform MAPI on $\mathbf{C}$ to obtain the dominant eigenvector $\mathbf{w}_1\in \mathbb{R}^{D^2\times 1}$;
		\STATE Perform MAPI on $\mathbf{C}-(\mathbf{w}_1\oplus\mathbf{w}_1^T)\mathbf{C}$ to obtain the second dominant eigenvector $\mathbf{w}_2\in \mathbb{R}^{D^2\times 1}$;
		\STATE $\mathbf{w} = [\mathbf{w}_1\ \mathbf{w}_2] \in \mathbb{R}^{D^2\times 2}$
		\STATE Reconstruct image:  $\mathbf{\hat{v}}=(\mathbf{w}\oplus\mathbf{w}^T)(\mathbf{v}_i-\mathbf{m})+\mathbf{m}$;
		\STATE Reshape $\mathbf{\hat{v}}$ back to the matrix form $\mathbf{\hat{I}}$;
		\RETURN reconstructed image $\mathbf{\hat{I}}$.\\
	\end{algorithmic}
\end{algorithm}

In Table \ref{tab: Image Restoration Results for Images corrupted by Noise Patches}, we compare the peak signal-to-noise ratio (PSNR) performances provided by different algorithms. As Fig.~\ref{fig: Image Reconstruction l2_3b} and~\ref{fig: Image Reconstruction min_3b} show, MAPI provides higher PSNRs than RPI, globally. The average PSNR of min2-PI, which represents MAPI via the min2 operator in (\ref{odotproduct2}), is 1.62 dB higher than RPI. The MAPI based reconstruction (24.8 dB) is also superior to recursive $\ell_1$-PCA (24.2 dB). Thus, based on the experimental PSNR results in Table 1, we can conclude that MAPI is as robust as $\ell_1$-PCA based restoration. We also expect the computational cost MAPI to be much lower than $\ell_1$-PCA methods. In fact, for an $N\times N$ covariance matrix, RPI needs $O(N^2)$ multiplications per iteration while MAPI avoids a quadratic number of multiplications and needs $O(N^2)$ additions only. Note that $\ell_1$-PCA utilizes regular power iteration updates and need $O(N^2)$ multiplications also. We also experimentally validated in Matlab that to compute the first principal vector, the recursive $\ell_1$-PCA and MAPI methods need 4.54 and 3.24 seconds, respectively. 

\begin{figure*}[htbp]
\centering
\subfloat[\label{fig: Image Reconstruction original_3b}]{
\begin{minipage}{0.1\linewidth}
\includegraphics[width=1\linewidth]{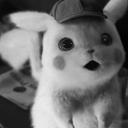}\\
\includegraphics[width=1\linewidth]{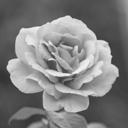}
\end{minipage}}
\subfloat[\label{fig: Image Reconstruction occluded_3b}]{
\begin{minipage}{0.1\linewidth}
\includegraphics[width=1\linewidth]{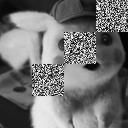}\\
\includegraphics[width=1\linewidth]{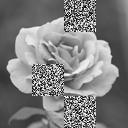}
\end{minipage}}
\subfloat[\label{fig: Image Reconstruction l1_3b}]{
\begin{minipage}{0.1\linewidth}
\includegraphics[width=1\linewidth]{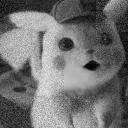}\\
\includegraphics[width=1\linewidth]{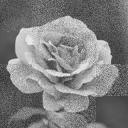}
\end{minipage}}
\subfloat[\label{fig: Image Reconstruction l2_3b}]{
\begin{minipage}{0.1\linewidth}
\includegraphics[width=1\linewidth]{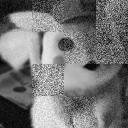}\\
\includegraphics[width=1\linewidth]{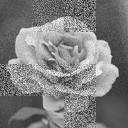}
\end{minipage}}
\subfloat[\label{fig: Image Reconstruction min_3b}]{
\begin{minipage}{0.1\linewidth}
\includegraphics[width=1\linewidth]{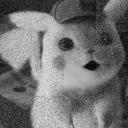}\\
\includegraphics[width=1\linewidth]{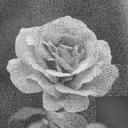}
\end{minipage}}
\subfloat[\label{fig: Image Reconstruction min2_3b}]{
\begin{minipage}{0.1\linewidth}
\includegraphics[width=1\linewidth]{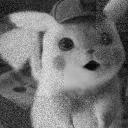}\\
\includegraphics[width=1\linewidth]{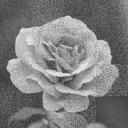}
\end{minipage}}\vspace{-5pt}
\caption{Image reconstruction examples: (a) original images; (b) occluded images; (c) recursive $
\ell_1$-PCA \cite{markopoulos2014optimal} results; (d) PCA using RPI results; (e) min1-PI results; (f) min2-PI results. MAPI is as robust as $\ell_1$-PCA based restoration.}
\vspace{-10pt}
\label{fig: Image Reconstruction}
\end{figure*}

\begin{table}[htbp]
	\centering
	\caption{Image Reconstruction PSNRs (dB). }
	\begin{tabular}{p{1.2cm}p{0.9cm}p{1.3cm}p{0.9cm}p{0.9cm}p{0.9cm}}
		\hline\noalign{\smallskip}
		\textbf{Image}&\textbf{Occlu- ded}&\textbf{Recursive $\ell_1$-PCA}&\textbf{RPI}&\textbf{Min1-PI}&\textbf{Min2-PI}\\
		\noalign{\smallskip}\hline\noalign{\smallskip}
        Statue 1&17.67&26.86&26.67&\textbf{27.12}&27.08\\
        Statue 2&16.87&25.02&24.61&\textbf{25.45}&25.34\\
        Earth&14.91&21.74&22.13&20.78&\textbf{25.03}\\
        Pikachu&15.29&\textbf{22.71}&18.86&22.70&22.58\\
        Flower&16.40&24.41&21.22&\textbf{24.59}&24.44\\
        Orange&15.77&23.69&23.59&22.50&\textbf{26.27}\\
        Cat&16.91&\textbf{24.82}&24.80&24.28&24.52\\
        Lenna&16.77&24.71&22.86&\textbf{25.59}&24.73\\
        Food&15.94&23.85&22.72&23.68&\textbf{23.70}\\
        Car&15.42&23.37&23.15&\textbf{23.43}&23.41\\
        Cobra&16.81&25.10&22.31&25.02&\textbf{25.57}\\
        River&17.27&25.23&24.50&25.07&\textbf{25.88}\\
        Butterfly&16.66&{\bf 24.91}&24.40&23.93&23.92\\
        Bridge&15.66&22.91&22.23&22.55&\textbf{24.34}\\
		\noalign{\smallskip}\hline\noalign{\smallskip}
        Average&16.31&24.24&23.15&24.05&\textbf{24.77}\\
		\noalign{\smallskip}\hline\noalign{\smallskip}
	\end{tabular}
	\label{tab: Image Restoration Results for Images corrupted by Noise Patches}\vspace{-20pt}
\end{table}

% Also, although MAPI yields a lower average PSNR than the recursive $\ell_1$-PCA and the min-PCA, it is much more efficient as it avoids the multiplication operation.

\subsection{Comparison with Stochastic PCA Method}
We now compare the MAPI with the stochastic power iteration method on a synthetic dataset as in \cite{xu2018accelerated}. The synthetic dataset $\smash{\mathbf{X}\in\mathbb{R}^{10^6\times10}}$ is generated using singular value decomposition. In detail, for a diagonal matrix $\mathbf{\Sigma}=\text{diag}\{1, \sqrt{0.9}, ..., \sqrt{0.9}\}\in\mathbb{R}^{10\times10}$, and random orthogonal  matrices $\mathbf{U}\in\mathbb{R}^{10^6\times10},\,\mathbf{V}\in\mathbb{R}^{10\times10}$, it is guaranteed that the matrix $\mathbf{C}=\mathbf{X}^T\mathbf{X}$ has an eigen-gap of $0.1$, where $\mathbf{X}=\mathbf{U\Sigma V}^T$. 

MAPI method can be also implemented into the mini-batch power method with momentum algorithm \cite[Algorithm 1]{xu2018accelerated}. We apply MAPI on 
$\mathbf{A}=\mathbf{X}^T\odot\mathbf{X}$ with the same mini-batch strategy, as Algorithm~\ref{al: Mini-batch Min Power Iteration with Momentum} shows. The difference between Algorithm~\ref{al: Mini-batch Min Power Iteration with Momentum} in this paper and Algorithm 1 in \cite{xu2018accelerated} is that we change all regular multiplication operations with our min1 operation (\ref{odotproduct}) and normalize the iterates using the $\ell_1$-norm. At the end of iterations we normalize the final vector using the  $\ell_2$-norm of the vector. 
In this experiment, we cannot use min2-PI operation because the min2 operation (\ref{odotproduct2}) cannot return negative entries. The vector $\mathbf{u}_1$ can have negative entries.

\begin{figure}[h]
    \centering
    \subfloat[\label{fig: stocastic_a}]{\includegraphics[width=0.5\linewidth]{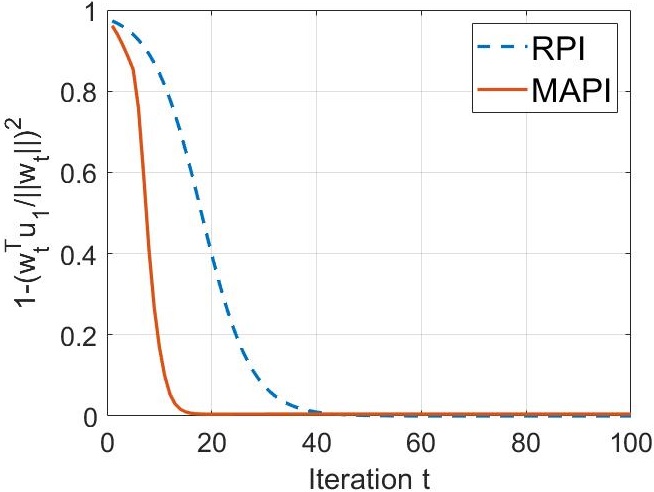}}
    \subfloat[\label{fig: stocastic_b}]{\includegraphics[width=0.5\linewidth]{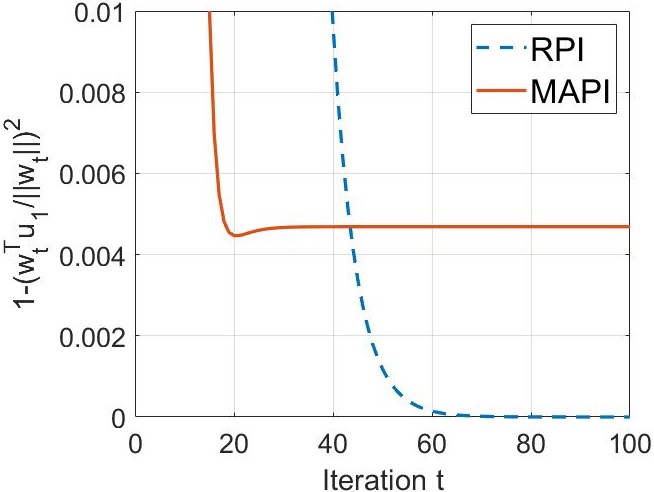}}
    \vspace{-10pt}\caption{RPI versus MAPI on a synthetic dataset $\mathbf{X}\in\mathbb{R}^{10^6\times10}$ where the covariance matrix has eigen-gap $\Delta=0.1$.}
    \label{fig: stocastic}\vspace{-10pt}
\end{figure}

In Fig.~\ref{fig: stocastic}, we compute $(1-(\mathbf{w}_t^T\mathbf{u}_1/||\mathbf{w}_t||)^2)$ at each iteration, where $\mathbf{u}_1$ is the dominant eigenvector obtained from eigendecomposition as in~\cite{xu2018accelerated}. 
The MAPI iteration converges to a vector very close to the actual eigenvector because 
 $(1-(\mathbf{w}_t^T\mathbf{u}_1/||\mathbf{w}_t||)^2)$ is very close to zero as shown in Fig.~\ref{fig: stocastic_a} for $t$ greater than 20. Such preservation of the eigenvector information can be considered as a direct analogue of Bussgang's theorem \cite{giunta1997bussgang, giunta1998bussgang}. However, $\mathbf{w}_T$ obtained using the MAPI method is not exactly the same as the eigenvector $\mathbf{v}_1$ as we see from Fig.~\ref{fig: stocastic_b} whose vertical axis has a different range from Fig.~\ref{fig: stocastic_a}. This is expected because we perform the iterations in the RKHS domain. Nevertheless, based on our observation, ranks of the entries of $\mathbf{w}_T$ obtained from the min1-PI
%(which entry is the largest, which is the second largest, etc.)  
are the same as the ranks obtained from $\mathbf{w}_T$ of the RPI. After 100 iterations, $\mathbf{w}_T=[$-0.143, 0.417, -0.117, 0.386, 0.129, -0.132, -0.433, -0.210, -0.577, -0.211] from the RPI and $\mathbf{w}_T=[$-0.165, 0.417, -0.137, 0.389, 0.148, -0.151, -0.431, -0.236, -0.536, -0.237] from the min1-PI. When we order from the largest to the smallest, both vectors produce the same ranks \{2, 4, 5, 3, 6, 1, 8, 10, 7, 9\}.
Hence, min1-PI converges significantly faster than RPI and both algorithms produce the same final ranks. In Section~\ref{secPageRank}, we will show how these properties allow us to employ MAPI in graph-based ranking algorithms such as Google PageRank \cite{page1999pagerank}. Moreover, if we need $(1-(\mathbf{w}_t^T\mathbf{u}_1/||\mathbf{w}_t||)^2)$ to reach ``absolute'' $0$, we can optimize the vector via the MAPI first and then switch to the RPI for further convergence.

\begin{algorithm}[H]
	\caption{Mini-batch MAPI with Momentum}
	\label{al: Mini-batch Min Power Iteration with Momentum}
	\begin{algorithmic}[1]
		\renewcommand{\algorithmicrequire}{\textbf{Input:}}
		\renewcommand{\algorithmicensure}{\textbf{Output:}}
		
		\REQUIRE Data $\mathbf{X}\in\mathbb{R}^{N\times D}$, iteration times $T$, batch size $s$, momentum parameter $\beta$.
% 		\ENSURE  Dominant-pseudo-eigenvector $\mathbf{w}_T \in \mathbb{R}^{N\times1}$.
		
		\STATE Initialize $\mathbf{w_0}$ as a random vector with $||\mathbf{w_0}||=1$;
		\FOR{$t=0, 1, \ldots, T-1$} 
		\STATE Generate a mini batch of independent and identically distributed samples $B=\{\tilde{\mathbf{A}}_{t_1}, ..., \tilde{\mathbf{A}}_{t_s}\}$
		\STATE $\mathbf{w}_{t+1} = (\frac{1}{s}\sum_{i=1}^s\tilde{\mathbf{A}}_{t_i})\odot \mathbf{w}_t-\beta \mathbf{w}_{t-1}$;
		\STATE $\mathbf{w}_{t} = \mathbf{w}_{t}/||\mathbf{w}_{t+1}||_1, \mathbf{w}_{t+1} = \mathbf{w}_{t+1}/||\mathbf{w}_{t+1}||_1$;
		\ENDFOR 
		\STATE (optional) $\mathbf{w}_T = \mathbf{w}_T/||\mathbf{w}_T||$;
		\RETURN  dominant pseudo-eigenvector $\mathbf{w}_T$.\\
	\end{algorithmic}
\end{algorithm}

\subsection{PageRank Algorithm Using MAPI}\label{secPageRank}
PageRank algorithm uses the hyperlink structure of the web to view inlinks into a page as a recommendation of that page from the author of the inlinking page \cite{benincasa2006page}. Specifically, it relies on applying power iteration to the matrix
\begin{equation}
    \mathbf{G} = \alpha \mathbf{H} + \frac{1-\alpha}{N}\mathbf{1}
\end{equation}
where $\mathbf{H}$ is the network adjacent matrix, $\mathbf{1}$ is an all-ones matrix, and $\alpha=0.85$ is known as a damping factor.

We first execute the PageRank algorithm with RPI or MAPI on the network graph in Fig.~\ref{fig: Network Graph}. We compute $||\mathbf{w}_t-\mathbf{w}_{t-1}||$ after each iteration to compare the convergence. 

\begin{figure}[htbp]\vspace{-10pt}
    \centering
    \subfloat[Network Graph\label{fig: Network Graph}]{\includegraphics[width=0.5\linewidth]{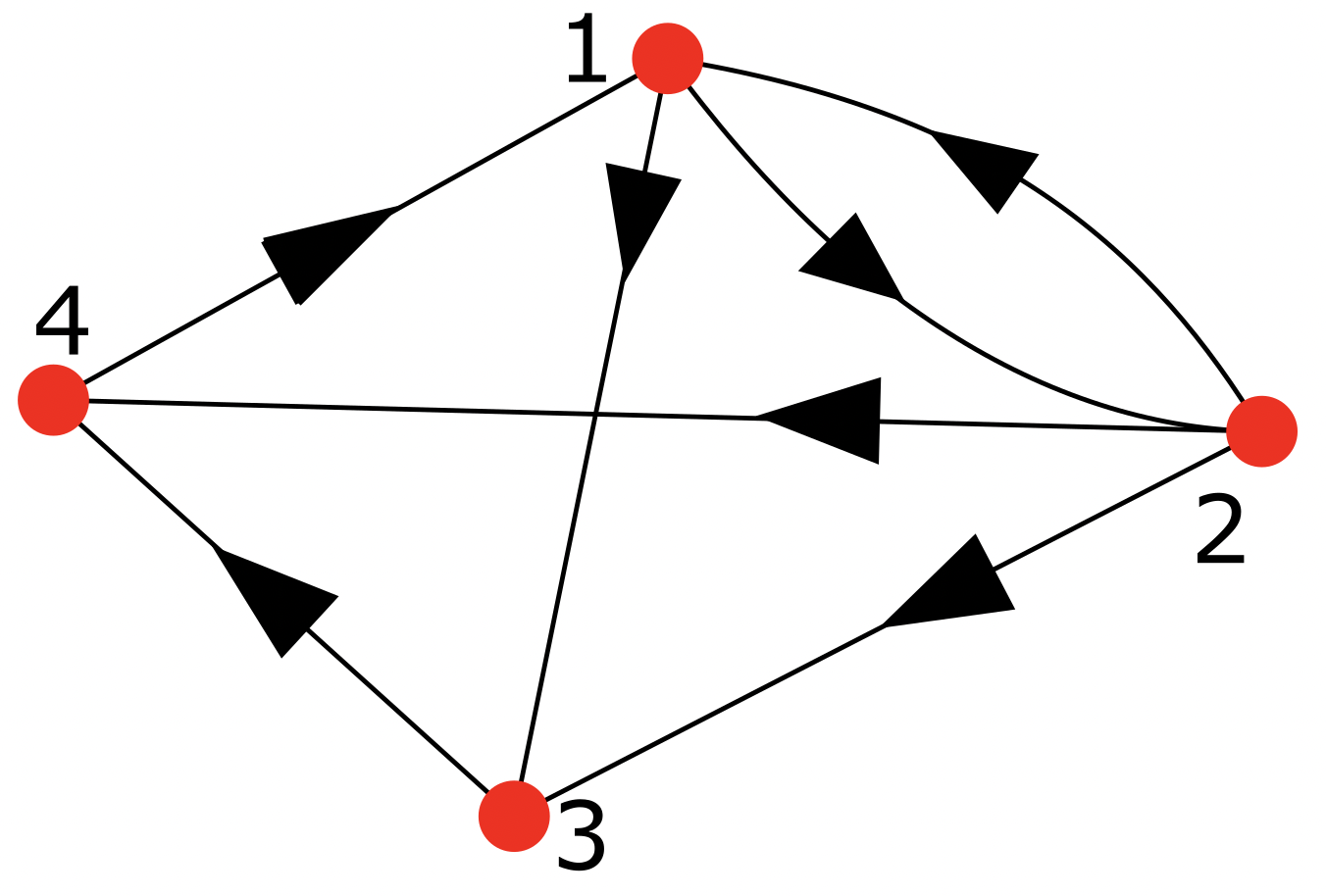}}
    \subfloat[Convergence Curve\label{fig: Convergence Curve }]{\includegraphics[width=0.5\linewidth]{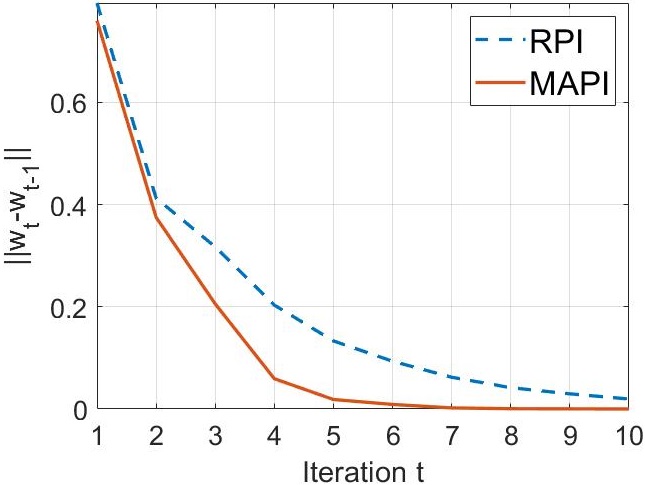}}
    \vspace{-5pt}\caption{A PageRank example.}
    \label{fig: PageRank}\vspace{-5pt}
\end{figure}

We apply one $\ell_2$-normalization after the final iteration for easier comparison. After 20 iterations,  RPI yields $w_T=[$0.634, 0.345, 0.434, 0.535], and MAPI provides $w_T=[$0.568, 0.349, 0.520, 0.535]. Consequently, although the weights are different, both power iteration methods provide the same index ranks \{1, 4, 3, 2\}.

%Moreover, we try the two methods on some large networks. However, because the network is very huge, it will take a long time to get converged. Therefore, here we only discuss the speed of convergence. As Figure~\ref{fig: PageRank_large} shows, our min power iteration converges remarkably faster than the conventional power iteration in the PageRank algorithm.

%\begin{figure}[H]
%    \centering
%    \subfloat[CA-GrQc\label{fig: CA_AstroPh}]{\includegraphics[width=0.5\linewidth]{CA_GrQc.jpg}}
%    \subfloat[CA-AstroPh\label{fig: CA_AstroPh}]{\includegraphics[width=0.5\linewidth]{CA_AstroPh.jpg}}
%    \caption{PageRank convergence curves on large networks. There are 5,242 nodes with 28,980 edges in CA-GrQc dataset, and there are 18,772 nodes with 396,160 edges in CA-AstroPh dataset.}
%    \label{fig: PageRank_large}
%\end{figure}

We further use the PageRank algorithm with both RPI and MAPI on Gnutella08 \cite{p2p-Gnutella08} and Gnutella09 \cite{p2p-Gnutella09} peer-to-peer network datasets. Gnutella08 contains 6,301 nodes with 20,777 edges, and Gnutella09 contains 8,114 nodes with 26,013 edges. The convergence curves are shown in Fig.~\ref{fig: Gnutella peer-to-peer network datasets}. MAPI still converges remarkably faster than RPI for both networks. However, because the size of the networks are very large, the ranks provided by the two methods are not the same. After 10 iterations, from Gnutella08, the top-10 ranks of the indices from the RPI are \{367, 249, 145, 264, 266, 123, 127, 122, 1317, 5\}, while the MAPI yields \{266, 123, 367, 127, 424, 249, 145, 264, 427, 251\}. On the other hand, from Gnutella09, the top-10 ranks of the indices from the RPI are \{351, 563, 822, 534, 565, 825, 1389, 1126, 356, 530\}, while the MAPI returns \{351, 822, 51, 1389, 563, 565, 530, 825, 356, 1074\}. Therefore, there are 7 common top-10 ranks of indices \{367, 249, 145, 265, 226, 123, 127\} from the 6,301-node dataset Gnutella08 and 8 common top-10 ranks of indices \{351, 563, 822, 565, 825, 1389, 356, 530\} from the 8,114-node dataset Gnutella09. Therefore, if we implement the MAPI in the page-rank-based web search system, the links displayed on the first page is very close to the conventional page-rank based system. The search time can also be reduced significantly. 
\vspace{-5pt}
\begin{figure}[htbp]
    \subfloat[Gnutella08\label{fig: CA_AstroPh}]{\includegraphics[width=0.5\linewidth]{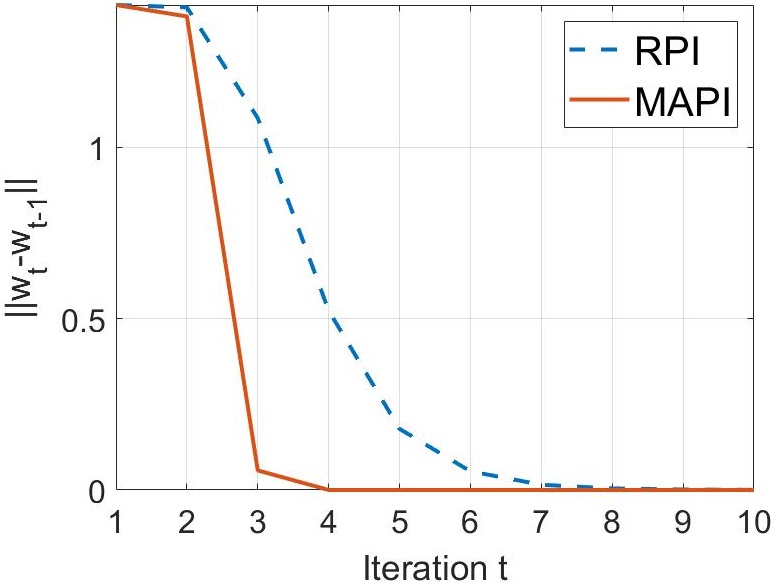}}
   \subfloat[Gnutella09\label{fig: CA_AstroPh}]{\includegraphics[width=0.5\linewidth]{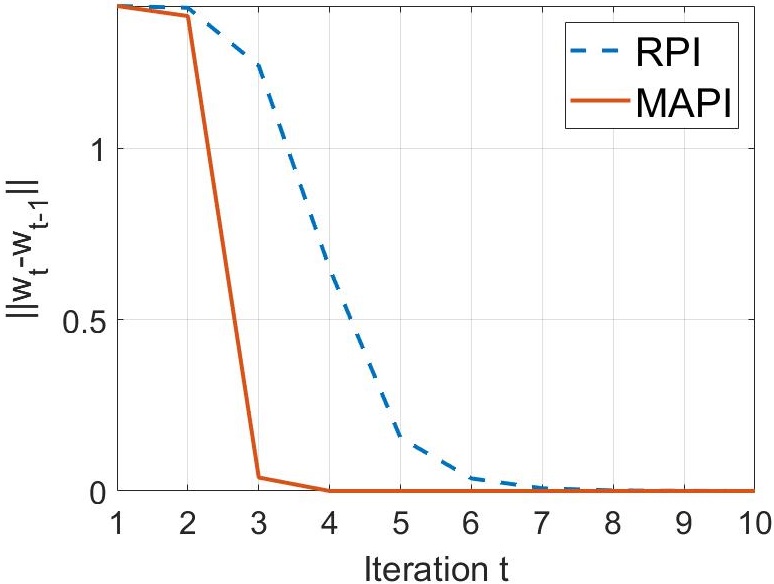}}\vspace{-5pt}
   \caption{PageRank convergence curves on Gnutella.}
   \label{fig: Gnutella peer-to-peer network datasets}
\end{figure}
\vspace{-20pt}
\section{Conclusion}
\label{secConclusions}
We proposed two types of multiplication-avoiding power iteration (MAPI) algorithms which replace the standard vector product operations in regular power iteration (RPI) with multiplication-free vector products. The MAPI is energy efficient because it significantly reduces the number of multiplication operations, which are known to be costly in terms of energy consumption in many processors. MAPI is as robust as $\ell_1$ principal component analysis as its kernels are related with the $\ell_1$-norm. According to our graph-based ranking experiment, the MAPI is also superior to the RPI in terms of convergence speed and energy efficiency. Though the final ranking obtained from the MAPI is not always exactly identical to the ranking obtained from the RPI but they are very close to each other. 
%Compared to the RPI, the min2-PI reaches an average PSNR 1.62 dB higher in our image reconstruction experiment. 
%The MAPI reduces the number of multiplications from  $n^2$ multiplications to $n$ multiplications per iteration for an $n\times n$ matrix.
%we can still implement the MAPI into the Google PageRank algorithm to accelerate the convergence because the ranks of the MAPI are roughly identical to the ranks of the RPI. In this case, the ranks are not exactly same but very close. 
%\vfill
%\pagebreak

% \vfill\pagebreak

\bibliographystyle{IEEEtran}
\bibliography{main}

\end{document}